\documentclass[10pt]{article}

\usepackage[margin=1in]{geometry}
\usepackage{amsmath,amssymb,amsthm}
\usepackage{booktabs,array}
\usepackage{microtype}
\usepackage{tikz}
\usetikzlibrary{arrows.meta,positioning}
\usepackage[hidelinks]{hyperref}
\usepackage{xurl}
\newtheorem{theorem}{Theorem}
\newtheorem{proposition}[theorem]{Proposition}
\newtheorem{lemma}[theorem]{Lemma}
\theoremstyle{remark}

\newcommand{\PoA}{\operatorname{PoA}}
\newcommand{\Q}{\mathbb Q}
\newcommand{\Ehat}{\widehat{\mathcal E}}

\title{An Exact Counterexample to Affine-Like\\
Price-of-Anarchy Shape in Quartic BPR Routing}
\author{Ian D'Ambrosio\\\small Nth Research Collective}
\date{9 August 2026}

\begin{document}
\maketitle

\begin{abstract}
We give an exact computer-assisted counterexample to a direct common-degree
quartic extension of the affine active-network shape theorem for
demand-dependent Price of Anarchy.  The instance is a directed network with
five vertices, six edges, and three origin--destination paths, all with
positive rational costs $c_e(x)=a_e+b_ex^4$.  Exact rational interval
certificates show that all three paths carry positive Wardrop flow throughout
the demand interval $[17,24]$, while
\[
 \PoA(21)>\PoA(17),\qquad \PoA(21)>\PoA(24).
\]
Continuity therefore forces an interior local maximum despite a constant
equilibrium active network.  Krawczyk inclusions isolate the Wardrop and
social-optimum KKT solutions, elementary boundary inequalities certify
constant support, and interval social costs certify both strict comparisons.
At differentiability points, we also derive a serial-edge identity explaining
how a common quartic edge can alter the derivative of PoA without changing
either route split.  The proof and certificate use exact rational arithmetic
for every interval and sign decision.
\end{abstract}

\section{Introduction}

Let a nonatomic population of demand $\mu>0$ choose routes from one origin to
one destination.  A Wardrop equilibrium uses only minimum-latency paths.  If
$C^{\rm eq}(\mu)$ and $C^{\rm opt}(\mu)$ are respectively the equilibrium and
minimum social costs, the Price of Anarchy is
\[
 \PoA(\mu)=\frac{C^{\rm eq}(\mu)}{C^{\rm opt}(\mu)}.
\]
Demand changes can alter both route flows and the set of shortest paths.
Cominetti, Dose, and Scarsini formalized the relevant breakpoints using the
\emph{active network} $\Ehat(\mu)$, the union of edges on all equilibrium
shortest paths, and proved a strong shape theorem for affine costs: between
consecutive changes of $\Ehat$, PoA is monotone or has one interior minimum
and no interior maximum \cite{CominettiDoseScarsini2024}.

That paper also gives a nonlinear counterexample with mixed degrees,
$c_1(x)=x$ and $c_2(x)=1+x^2$.  It does not answer whether the affine shape
survives for costs in the standard common-degree quartic BPR form.  This
distinction matters
because common-degree costs obey a scaling relation between equilibrium and
the social optimum \cite{OHareConnorsWatling2016}.  Dose's subsequent
two-link polynomial analysis concerns demands at which equilibrium is optimal
and explicitly leaves polynomial shape results and larger networks open
\cite{Dose2026}.  Related smoothness results require strictly positive edge
derivatives and therefore do not apply at configurations containing
zero-loaded BPR edges, since $c'_e(0)=0$
\cite{CominettiDoseScarsiniPhase2024}.

To our knowledge, the example below is the first exact counterexample to this
direct common-degree quartic extension.  The counterexample is small, rational,
and exactly checkable.  Its PoA bump is numerically shallow, but its strictness
does not depend on floating-point arithmetic.  No claim of topology minimality
or empirically calibrated coefficients is made.

\section{Model and optimality conditions}

For each directed edge $e$, let
\[
 c_e(t)=a_e+b_et^4,\qquad a_e,b_e\in\Q_{>0}.
\]
For path flows $f$ and induced edge loads $\ell_e(f)$, Wardrop equilibria are
the minimizers of the Beckmann potential
\cite{BeckmannMcGuireWinsten1956}
\[
 \Phi(f)=\sum_e\left(a_e\ell_e(f)+\frac{b_e}{5}\ell_e(f)^5\right)
\]
over the demand simplex.  Social optima minimize
\[
 C(f)=\sum_e\left(a_e\ell_e(f)+b_e\ell_e(f)^5\right).
\]
Thus their path KKT equations use the marginal edge costs
$a_e+5b_e\ell_e^4$.  Positivity of the $b_e$ makes both objectives strictly
convex in edge loads.  In the network below, the path--edge incidence matrix
has full column rank, so the path-flow minimizers are unique and continuous in
demand.

\section{The counterexample}

Consider the directed graph in Figure~\ref{fig:network}.  Edge 0 is common to
every route, and edges 1--5 form a Wheatstone network.  There are exactly three
simple origin--destination paths:
\[
 P_x=(0,1,2),\qquad P_y=(0,3,4),\qquad P_z=(0,1,5,4).
\]

\begin{figure}[ht]
\centering
\begin{tikzpicture}[
  node distance=18mm and 25mm,
  vertex/.style={circle,draw,inner sep=1.5pt,minimum size=7mm},
  edge/.style={-{Latex[length=2mm]},thick}
]
  \node[vertex] (s) {$s$};
  \node[vertex,right=of s] (a) {$a$};
  \node[vertex,above right=of a] (u) {$u$};
  \node[vertex,below right=of a] (v) {$v$};
  \node[vertex,right=32mm of a] (t) {$t$};
  \draw[edge] (s) -- node[above] {$0$} (a);
  \draw[edge] (a) -- node[above left] {$1$} (u);
  \draw[edge] (u) -- node[above right] {$2$} (t);
  \draw[edge] (a) -- node[below left] {$3$} (v);
  \draw[edge] (v) -- node[below right] {$4$} (t);
  \draw[edge] (u) -- node[right] {$5$} (v);
\end{tikzpicture}
\caption{The common-edge Wheatstone network.}
\label{fig:network}
\end{figure}
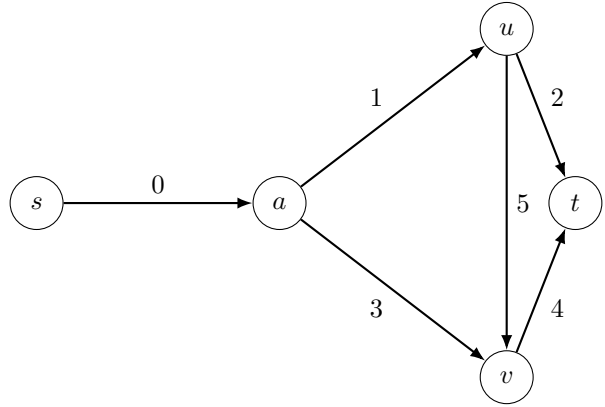

The coefficient pairs are given in Table~\ref{tab:coefficients}.

\begin{table}[ht]
\centering
\caption{Positive rational coefficients in $c_e(t)=a_e+b_et^4$.}
\label{tab:coefficients}
\begin{tabular}{@{}clrr@{}}
\toprule
$e$ & role & $a_e$ & $b_e$ \\
\midrule
0 & common serial & $12{,}800{,}000$ & $1855/374$ \\
1 & upper left   & $17/3$   & $13/896$ \\
2 & upper right  & $236$    & $38/5$ \\
3 & lower left   & $1280$   & $17/80$ \\
4 & lower right  & $5504/7$ & $188$ \\
5 & cross edge   & $424/7$  & $208/3$ \\
\bottomrule
\end{tabular}
\end{table}

\begin{theorem}\label{thm:counterexample}
For the network and costs above, all three paths are active at Wardrop
equilibrium for every $\mu\in[17,24]$.  Nevertheless, $\PoA$ has an interior
maximum in $(17,24)$.
\end{theorem}

\subsection{Constant Wardrop support}

Let $(x,y,z)$ be the flows on $(P_x,P_y,P_z)$, with
$x+y+z=\mu$.  Excluding the common edge, the five loads are
\[
 (x+z,x,y,y+z,z).
\]
When all paths are used, equality of $P_x$ and $P_z$, and of $P_y$ and
$P_z$, is equivalent to
\begin{align}
 F_\mu(x,y)
 &=c_2(x)-c_5(z)-c_4(\mu-x)=0,\label{eq:F}\\
 G_\mu(x,y)
 &=c_3(y)-c_1(\mu-y)-c_5(z)=0,\label{eq:G}
\end{align}
where $z=\mu-x-y$.  Put
$\mathbf H_\mu=(F_\mu,G_\mu)$.  This polynomial map is continuously
differentiable.  For a rational box $X$, the verifier evaluates an interval
matrix $\mathbf H_\mu'(X)$ that contains every Jacobian
$\mathbf H_\mu'(w)$ for $w\in X$, and applies the two-dimensional Krawczyk
operator \cite{MooreKearfottCloud2009}
\[
 K(m,X)=m-C\mathbf H_\mu(m)
 +(I-C\mathbf H_\mu'(X))(X-m),
\]
where $m$ is the midpoint of $X$ and $C$ is the exact inverse of the midpoint
Jacobian.  The Krawczyk inclusion theorem says that strict inclusion
$K(m,X)\subset\operatorname{int}X$ proves a unique root in $X$.
At demand 17 it gives
\begin{align*}
11.745703474&<x<11.745703476,\\
 3.439049651&<y< 3.439049654,\\
 1.815246870&<z< 1.815246875.
\end{align*}
In particular, the unique equilibrium is in the simplex interior.

It remains to show that this branch cannot meet the boundary before demand
24.  The three possible boundary cases admit short exact exclusions.

\begin{lemma}\label{lem:boundary}
No solution of \eqref{eq:F}--\eqref{eq:G} with nonnegative path flows has
$17\leq\mu\leq24$ and $xyz=0$.
\end{lemma}

\begin{proof}
If $x=0$, then
\[
 F_\mu(0,y)\leq a_2-a_5-a_4=-\frac{4276}{7}<0,
\]
so \eqref{eq:F} cannot hold.

If $y=0$, equation \eqref{eq:G} forces $z<3$: for $z\geq3$,
$c_1(x+z)+c_5(z)>c_3(0)$.  Equation \eqref{eq:F} then reads
\[
 b_2x^4=a_5+a_4-a_2+(b_5+b_4)z^4.
\]
Direct substitution at $x=8,z=3$ shows that its left side is already larger
than the right side.  Hence $x<8$ and $\mu=x+z<11$.

If $z=0$, equations \eqref{eq:F}--\eqref{eq:G} are linear in
$X=x^4$ and $Y=y^4$.  Their unique solution is
\[
 X=\frac{3831364840}{18693}>21^4,
 \qquad
 Y=\frac{1083769565}{130851}>9^4.
\]
Thus $\mu=x+y>30$.
\end{proof}

\begin{lemma}[Support continuation]\label{lem:continuation}
The unique Wardrop path flow is continuous on $[17,24]$.  Starting from the
interior root at demand 17, it cannot leave the simplex interior before demand
24.
\end{lemma}

\begin{proof}
Write a feasible path flow as $f=\mu q$, where $q$ belongs to the fixed compact
unit simplex.  The Beckmann objective is continuous in $(\mu,q)$, and strict
convexity together with the full-column-rank incidence matrix makes its
minimizer unique.  Berge's maximum theorem therefore makes $q(\mu)$, and hence
$f(\mu)$, continuous \cite{Berge1963}.

If a coordinate first reached zero at some $\mu_*\in(17,24]$, choose demands
$\mu_n\uparrow\mu_*$ before that first contact.  All three paths are used at
$\mu_n$, so \eqref{eq:F}--\eqref{eq:G} hold there.  Continuity of the path
flows and costs preserves both equalities in the limit at $\mu_*$.  This would
give a nonnegative boundary solution forbidden by
Lemma~\ref{lem:boundary}.
\end{proof}

All three path flows are consequently positive throughout $[17,24]$, so all
three simple paths are shortest.  In the active-network convention,
$\Ehat(\mu)$ is the full six-edge graph throughout this interval.

\subsection{Exact PoA comparisons}

The certificate isolates the equilibrium KKT roots at demands 17, 21, and 24.
For the social optimum, it isolates an interior marginal-cost KKT root at 17.
At 21 and 24 the optimum uses $P_x$ and $P_y$ only; exact univariate brackets
for the $P_x$ flow are
\[
14.499119464<u_{21}<14.499119467,
\qquad
16.569846371<u_{24}<16.569846374.
\]
The marginal cost of the unused path $P_z$ is bounded strictly above the
common used-path marginal cost in both boxes.

Substituting these rational flow boxes into
\[
 C=\sum_e(a_e\ell_e+b_e\ell_e^5)
\]
gives rational intervals for all six social costs.  Dividing outward yields
the certified enclosures in Table~\ref{tab:poa}.  The displayed decimals are
coarser outward roundings of the exact rational endpoints.

\begin{table}[ht]
\centering
\caption{Exact-rational interval enclosures for PoA.}
\label{tab:poa}
\begin{tabular}{@{}cc@{}}
\toprule
$\mu$ & certified enclosure \\
\midrule
17 & $[1.0000089877,\;1.0000089879]$ \\
21 & $[1.0000100959,\;1.0000100961]$ \\
24 & $[1.0000088727,\;1.0000088730]$ \\
\bottomrule
\end{tabular}
\end{table}

The exact checker does not compare rounded decimals.  It verifies the positive
cross products
\begin{align*}
 \underline C^{\rm eq}_{21}\,\underline C^{\rm opt}_{17}
 -\overline C^{\rm eq}_{17}\,\overline C^{\rm opt}_{21}&>0,\\
 \underline C^{\rm eq}_{21}\,\underline C^{\rm opt}_{24}
 -\overline C^{\rm eq}_{24}\,\overline C^{\rm opt}_{21}&>0.
\end{align*}
Hence $\PoA(21)$ strictly exceeds the values at both endpoints.  PoA is
continuous, so its maximum on $[17,24]$ is attained in $(17,24)$.  This proves
Theorem~\ref{thm:counterexample}.

\section{A serial-edge mechanism}

The counterexample was found by exploiting an exact cancellation.  Let $H$ be
a base network, and fix an open demand interval $I$ on which its equilibrium
cost $\lambda$ and its optimum social cost $O$ are differentiable.  Write the
equilibrium social cost as $E(\mu)=\mu\lambda(\mu)$ and the optimum marginal
cost as $M(\mu)=O'(\mu)$.  Add one edge common to every path, with latency
$A+B\mu^4$.  Its social-cost contribution is
\[
 G(\mu)=A\mu+B\mu^5.
\]
Because every feasible route crosses the new edge, neither the equilibrium nor
the optimum route split changes.

\begin{proposition}\label{prop:mechanism}
For every $\mu\in I$, let $N_0=E'O-EO'$ be the numerator of the derivative of
$E/O$.  After adding the common edge, the new derivative numerator is
\[
 N=N_0+A\alpha+B\beta,
\]
where
\[
 \alpha=\mu^2\lambda'+O-\mu M,
 \qquad
 \beta=\mu^4\bigl(\alpha+4(O-\mu\lambda)\bigr).
\]
\end{proposition}

\begin{proof}
Expand
\[
 (E'+G')(O+G)-(E+G)(O'+G').
\]
The quadratic terms $G'G-GG'$ cancel.  Collecting the coefficients of $A$
and $B$, then using $E=\mu\lambda$, gives the stated formulas.
\end{proof}

At two differentiability points, requiring $N>0$ at the first and $N<0$ at the
second is therefore a two-variable linear-feasibility problem in positive
$(A,B)$.  We used this observation as a numerical discovery heuristic to tune
a Wheatstone base and common edge; the discovery trajectory is not part of the
proof.  The common edge leaves its Braess-type path migration untouched while
reweighting the derivative through $A$ and $B$.  This explains why untargeted
coefficient sampling can miss the effect: the route split and the PoA
normalization must be tuned separately.

\section{Reproducibility and scope}

The archival supplement contains the paper source, canonical certificate,
standalone exact checker, SHA-256 manifest, and prior-art record, following
standard reproducibility practice \cite{StoddenEtAl2016}.  The exact version
1.0.0 archive is permanently available at
\href{https://doi.org/10.5281/zenodo.21864490}{doi:10.5281/zenodo.21864490}.
After extracting the supplement, run
\begin{verbatim}
python3 verify_audit.py
\end{verbatim}
This one command checks every manifest entry, runs
\path{ancillary/verify_counterexample.py} in a fresh subprocess, and compares
its output byte-for-byte with \path{ancillary/certificate.json}.  The checker
uses Python's \texttt{Fraction} type for every interval and sign decision;
floating-point values appear only as convenience fields after the
corresponding exact signs have been checked.  The certificate's expected
SHA-256 digest is
\begin{center}
\small\ttfamily
1234e8feb90a19dd7dfd29f82cf92af5cd32ec4141e82d1a4f4d4843c05b7eb6.
\end{center}

Theorem~\ref{thm:counterexample} disproves the direct common-degree quartic
extension described in the introduction.  It does not determine the shape for
parallel or series-parallel quartic networks, bound the possible number of
extrema, or establish that the topology or coefficient sizes are minimal.  The
large engineered common free-flow term makes the certified PoA variation
small; it does not affect the exact sign comparisons.

\paragraph{Research-system disclosure.}
This research was conducted by Ian D'Ambrosio with assistance from Beyond, an
AI-assisted research system operated by Nth Research Collective.  Beyond
supported literature retrieval, hypothesis generation, proof critique, code
and experiment design, orchestration, adversarial checking, and drafting.
Ian D'Ambrosio made all final scientific judgments, verified the cited sources
and computational evidence, and accepts full responsibility for the accuracy,
integrity, originality, and conclusions.  Beyond is not an author and cannot
approve the manuscript or accept responsibility for the work.

\begingroup
\small
\bibliographystyle{plain}
\bibliography{references}
\endgroup

\end{document}